\documentclass[a4paper,10pt]{eptcs}
\usepackage[utf8]{inputenc}
\usepackage[english]{babel}
\usepackage{amsmath,amsthm,amsfonts,amssymb,stmaryrd}
\renewcommand{\marginpar}[2][]{%
}
\usepackage[all]{xy}
\CompileMatrices
\SelectTips{cm}{12} 
\newdir{ >}{{}*!/-5pt/@{>}}
\usepackage{tikz}
\usetikzlibrary{backgrounds}
\usepackage{ifthen}

\newcommand{\mnode}[3][]{%
  \ifthenelse{\equal{#1}{}}%
  {\node[inner sep = 2pt] (#2) at (#3) {${#2}$};}%
  {\node[inner sep = 2pt] (#1) at (#3) {${#2}$};}%
}

\usepgflibrary{arrows}
\usetikzlibrary{topaths}
\newcommand{\ardraw}[3][->]{
  \ardrawL[#1]{#2}{#3}{ }{ };
}

\newcommand{\ardrawL}[5][]{%
    \ifthenelse{\equal{#4}{}}
    {%
      \ifthenelse{\equal{#1}{}}
      {%
        \begin{scope}[inner sep=1.5pt,fill=white]
          \draw[->] (#2) --node[#4,fill]  {$\scriptstyle#5$} (#3);
        \end{scope}%
      }%
      {%
        \begin{scope}[inner sep=1.5pt,fill=white]
          \draw[->]  (#2) #1   node[#4,fill]  {$\scriptstyle#5$} (#3);
        \end{scope}%
      }%
    }
    {%
      \ifthenelse{\equal{#1}{}}
      {%
        \begin{scope}[inner sep=1.5pt]
          \draw[->] (#2) --node[#4]  {$\scriptstyle#5$} (#3);
        \end{scope}
      }%
      {%
        \begin{scope}[inner sep=1.5pt]
          \draw[->] (#2)  #1  node[#4]  {$\scriptstyle#5$} (#3);
        \end{scope}
      }%
    }
}

\newcommand{\mardrawL}[5][--]{{%
    \ifthenelse{\equal{#4}{}}
    {\begin{scope}[inner sep=3pt,fill=white,fill]
      \draw[angle 90 reversed->] (#2) #1 node[#4,fill=white]  {$\scriptstyle#5$} (#3);
    \end{scope}}
    {\begin{scope}[inner sep=3pt]
      \draw[angle 90 reversed->] (#2) #1 node[#4]  {$\scriptstyle#5$} (#3);
    \end{scope}}
}}

\newcommand{\rardrawL}[5][--]{{%
    \ifthenelse{\equal{#4}{}}
    {\begin{scope}[inner sep=3pt,fill=white,fill]
      \draw[left hook->] (#2) #1 node[#4,fill=white]  {$\scriptstyle#5$} (#3);
    \end{scope}}
    {\begin{scope}[inner sep=3pt]
      \draw[left hook->] (#2) #1 node[#4]  {$\scriptstyle#5$} (#3);
    \end{scope}}
}}

\newcommand{\RardrawL}[5][--]{{%
    \ifthenelse{\equal{#4}{}}
    {\begin{scope}[inner sep=3pt,fill=white,fill]
      \draw[right hook->] (#2) #1 node[#4,fill=white]  {$\scriptstyle#5$} (#3);
    \end{scope}}
    {\begin{scope}[inner sep=3pt]
      \draw[right hook->] (#2) #1 node[#4]  {$\scriptstyle#5$} (#3);
    \end{scope}}
}}

\newcommand{\ardrawd}[3][--]{%
  {\tikzstyle{every node}=[inner sep=3pt]
    \draw[-,thick,white,double] (#2) #1  (#3);
    \draw[->] (#2) #1  (#3);}}

\newcommand{\pardrawL}[5][--]{{%
    \ifthenelse{\equal{#4}{}}
    {\begin{scope}[inner sep=3pt,fill=white,fill]
    \draw[>=left to,->] (#2) #1 node[#4,fill=white]  {$\scriptstyle#5$} (#3);
    \end{scope}}
    {\begin{scope}[inner sep=3pt]
    \draw[>=left to,->] (#2) #1 node[#4]  {$\scriptstyle#5$} (#3);
    \end{scope}}
}}

\newcommand{\POCS}[4][.09]{
  \begin{scope}[scale=#1]
  \draw (#2)++([scale=1]#4)++([scale=1]#3) 
  ++(#3) ++(#3) -- ++ (#4) -- ++ (#4) ;
  \draw (#2)++([scale=1]#3) ++([scale=1]#4) 
  ++(#4) ++(#4) -- ++ (#3)-- ++ (#3) ;
  \end{scope}
}

\newcommand{\POC}[3][.11]{%
  \POCS[#1]{#2}{#3-45:1}{#3+45:1}
}

\newcommand{\arthree}[2]{%
  \arthree{#1}{#2}{}{}%
}

\newcommand{\eardrawL}[5][--]{{%
    \ifthenelse{\equal{#4}{}}
    {\begin{scope}[inner sep=2pt,fill=white,fill]
      \draw[->>] (#2) #1 node[#4,fill=white]  {$\scriptstyle#5$} (#3);
    \end{scope}}
    {\begin{scope}[inner sep=2pt]
      \draw[->>] (#2) #1 node[#4]  {$\scriptstyle#5$} (#3);
    \end{scope}}
}}

\newcommand{\setN}{\mathbb{N}}

\DeclareMathOperator{\A}{\mathrm{ar}}
\DeclareMathOperator{\cnct}{\mathrm{cnct}}

\renewcommand{\bar}{\overline}

\newcommand{\R}{\mathcal{R}}
\renewcommand{\S}{\mathcal{S}}

\newcommand{\redlabel}[3]{#1 \rightarrow #2 \leftarrow #3}
\newcommand{\context}[3]{#1 \rightarrow #2 \leftarrow #3}
\newcommand{\redrule}[3]{#1 \leftarrow #2 \rightarrow #3}
\newcommand{\fullaction}[5]{(#2 \rightarrow #1) \xrightarrow{#2 \rightarrow #3 \leftarrow #4} (#4 \rightarrow #5)}
\newcommand{\action}[5]{(#2 \rightarrow #1) \xrightarrow{#2 \rightarrow #3 \leftarrow #4} (#4 \rightarrow #5)}
\newcommand{\actionLight}[5]{#2 \rightarrow #1 \xrightarrow{#2 \rightarrow #3 \leftarrow #4} #4 \rightarrow #5}
\newcommand{\theaction}{\action{G}{J}{F}{K}{H}}
\newcommand{\theactionLight}{\actionLight{G}{J}{F}{K}{H}}
\renewcommand{\ar}{\mathtt{ar}}

\newcommand{\pair}[2]{#1 \!\!\!\Join\!\!\! #2}
\newcommand{\compl}[2]{\widehat #1^{#2}}     
\newcommand{\minint}[2]{J^{^{_ #2}}_{^#1}} 

\let\oldSlashL\L
\renewcommand{\L}{\ensuremath{\text{\oldSlashL}}}

\usepackage{comment}
\usepackage[center]{subfigure}
\usepackage[small,bf,center]{caption}

\theoremstyle{plain}
\newtheorem{theorem}{Theorem}[section]
\newtheorem{lemma}[theorem]{Lemma}

\newtheorem{proposition}[theorem]{Proposition} 

\newtheorem{fact}[theorem]{Fact}
\theoremstyle{remark}
\newtheorem{remark}{Remark}[section]

\theoremstyle{definition}
\newtheorem{definition}[theorem]{Definition}
\newtheorem{example}{Example}[section]

\newcommand{\centerparagraph}[1]{\textsc{#1}}

\usepackage{ifthen}
\let\oldTo\to
\let\oldGets\gets
\renewcommand{\to}[1][ ]{%
  \ifthenelse{\equal{#1}{ }}
  {\oldTo}
  {\ensuremath{\mathrel{{\relbar}\mkern-1mu\raisebox{.3ex}{\ensuremath{\scriptstyle #1}}\mkern-1mu{\shortrightarrow}}}}
}
\renewcommand{\gets}[1][ ]{%
  \ifthenelse{\equal{#1}{ }}
  {\oldGets}
  {\ensuremath{\mathrel{{\shortleftarrow}\mkern-1mu\raisebox{.3ex}{\ensuremath{\scriptstyle #1}}\mkern-1mu{\relbar}}}}
}

\title{Structured Operational Semantics for Graph Rewriting\thanks{%
This work was partially supported by grants from Agence Nationale de la Recherche, ref. ANR-08-BLANC-0211-01 (COMPLICE project) and ref. ANR-09-BLAN-0169 (PANDA project).
}}

\author{Andrei Dorman
\institute{Dip. di Filosofia, Università Roma Tre \\ LIPN -- UMR\,7030, Universit\'e Paris 13}
\email{andrei.dorman@lipn.univ-paris13.fr} \and
Tobias Heindel 
\institute{LIPN -- UMR\,7030, Universit\'e Paris 13 }
\email{tobias.heindel@lipn.univ-paris13.fr}}

\providecommand{\titlerunning}{SOS for Graph Rewriting}
\providecommand{\authorrunning}{Andrei Dorman, Tobias Heindel}

\usepackage{microtype}
\begin{document}
\maketitle
\begin{abstract}
  Process calculi and graph transformation systems provide models of
  reactive systems with labelled transition semantics.  While the
  semantics for process calculi is compositional, this is not the case
  for graph transformation systems, in general.  Hence, the goal of
  this article is to obtain a compositional semantics for graph
  transformation system in analogy to the structural operational
  semantics (SOS) for Milner's Calculus of Communicating Systems (CCS).

  The paper introduces an SOS style axiomatization of the standard
  labelled transition semantics for graph transformation systems.  The
  first result is its equivalence with the so-called Borrowed Context
  technique.  Unfortunately, the axiomatization is not compositional
  in the expected manner as no rule captures ``internal''
  communication of sub-systems.  The main result states that such a
  rule is derivable if the given graph transformation system enjoys a
  certain property, which we call ``complementarity of actions''.
  Archetypal examples of such systems are interaction nets.  We also
  discuss problems that arise if ``complementarity of actions'' is
  violated.
\end{abstract}

\section{Introduction}
\label{sec:introduction}
Process calculi remain one of the central tools for the description of 
interactive systems. 
The archetypal example of  process calculi 
are Milner's $\pi$-calculus and the even more basic calculus of communication systems~\textsc{(ccs)}.
The semantics of these calculi is given by 
labelled transition systems (\textsc{lts}), 
which in fact can be given as a structural operational semantics (\textsc{sos}). 
An advantage of \textsc{sos} is their potential for combination with compositional methods for
the verification of systems (see e.g.~\cite{Simpson2004287}).

Fruitful inspiration for the development of \textsc{lts} semantics for other
``non-standard'' process calculi 
originates from the area of graph transformation
where techniques for the derivation of \textsc{lts} semantics from ``reaction rules'' 
have been developed~\cite{sassonesobocinski:njc,EH06}. 
The strongest point of these techniques is the context independence of the 
resulting behavioral equivalences, 
which are in fact congruences. 
Moreover, 
these techniques have lead to original~\textsc{lts}-semantics
for the ambient calculus~\cite{Rathke2010,bonchi2009labelled}, 
which are also given as \textsc{sos} systems. 
Already in the special case of ambients, 
the \textsc{sos}-style presentation goes beyond 
the standard techniques of label derivation in~\cite{sassonesobocinski:njc,EH06}. 
An open research challenge is the development of a general technique for the canonical derivation 
of \textsc{sos}-style~\textsc{lts}-semantics. 
The problem is the ``monolithic'' character of the standard \textsc{lts} for graph transformation systems.

In the present paper, we set out to develop a partial solution to the problem 
for what we shall call \emph{\textsc{ccs}-like} graph transformation systems. 
The main idea is to develop an analogy to \textsc{ccs} 
where each action $\alpha$ has a co-action $\bar \alpha$
that can synchronize to obtain a silent transition; 
this is the so-called \emph{communication rule}. 
In analogy, 
one can restrict attention to graph transformation systems 
with rules that 
allow to assign to each (hyper-)edge a unique \emph{co-edge}. 
Natural examples of such systems are interaction nets as introduced by Lafont\marginpar{deterministic?}~\cite{Laf90,Alex99}.
In fact, 
one of the motivations of the paper is to derive
\textsc{sos} semantics for interaction nets. 

\paragraph{Structure and contents of the paper}
We first introduce the very essentials of graph transformation 
and the so-called  Borrowed Context \textsc{(bc)} technique~\cite{EH06} 
for the special case of (hyper-)graph transformation
in Section~\ref{sec:preliminaries}.
To make the analogy between \textsc{ccs} and \textsc{bc}
as formal as possible, 
we introduce the system \textsc{sosbc}
in Section~\ref{sec:SOSsemantics}, 
which is meant to provide the uninitiated reader with
a new perspective on the \textsc{bc} technique.
Moreover, 
the system \textsc{sosbc} emphasizes the ``local'' character of graph transformations
as every transition can be decomposed into a ``basic'' action
in some context. 
In particular, we do not have any counterpart to the communication rule of \textsc{ccs}, 
which shall be addressed in Section~\ref{sec:interaction-rule}.
We illustrate why it is not evident when and how two labeled transitions
of two states that share their interface can be combined into a single
synchronized action. 
However, we will be able to describe sufficient conditions 
on (hyper-)graph transformation systems
that allow to derive the counterpart 
of the communication rule of \textsc{ccs} in the system \textsc{sosbc}. 
Systems of this kind have a natural notion of  ``complementarity of actions''
in the \textsc{lts}. 

\section{Preliminaries}
\label{sec:preliminaries}
 We first recall the standard definition of (hyper-)graphs and 
a formalism of transformation of
hyper-graphs (following the double pushout approach). 
We also present the 
labelled transition semantics for
hyper-graph transformation systems
that has been proposed in~\cite{EH06}.
{
In the present paper, 
the more general case of categories of graph-like structures
is not of central importance. 
However, 
some of the proofs will use basic results of category theory. 
}

\newcommand{\hg}[1][]{%
  (E_{#1},V_{#1},\ell_{#1} , \cnct_{#1} )%
}
\newcommand{\hgm}[1]{%
  (#1_E,#1_V)%
}
\begin{definition}[Hypergraphs and hypergraph morphisms]
  Let $\Lambda$ be a set of \emph{labels}
  with associated \emph{arity} function $\A \colon \Lambda \to \setN$. 
  A \emph{($\Lambda$-labelled) hyper-graph} is a tuple
  $G = \hg$
  where $E$ is a set of \emph{(hyper-)edges}, 
  $V$ is a set of \emph{vertices} or \emph{nodes},
  $\ell \colon E \to \Lambda$ is the  \emph{labelling function}, 
  and $\cnct$ is the \emph{connection} function, 
  which assigns to each edge $e \in E$ 
  a string (e.g. a finite sequence) of \emph{incident vertices} $\cnct(e) = v_1\cdots v_n$ 
  of length $\A (\ell(e))=n$ (where $\{v_1,\dots,v_n\}\subseteq V$). 
  Let $v \in V$ be a node;
  its \emph{degree}, written $\deg(v)$
  is the number of edges of which it is an incident node, 
  i.e.\ $\deg(v) = |\{ e \in E \mid v \text{ incident to } e \}|$
  (where for any finite set $M$, the number of elements of $M$ is $|M|$).
  We also write $v \in G$ and $e \in G$ if  $v \in V$ and $e \in E$. 

  Let $G_i = \hg[i]$ ($i \in \{ 1,2\}$)
  be hyper-graphs;
  a \emph{hyper-graph morphism} from $G_1$ to $G_2$, 
  written $f \colon G_1 \to G_2$
  is a pair of functions $f = (f_E \colon E_1 \to E_2 , f_V \colon V_1 \to V_2)$
  such that $\ell_2 \circ f_E = \ell_1$
  and for each edge $e_1 \in E_1$ with attached nodes
  $\cnct (e) = v_1 \cdots v_n$ we have
  $\cnct_2 (f_E(e)) = f_V(v_1) \cdots f_V(v_n)$.
  A hyper-graph morphism $f = \hgm{f} \colon G_1 \to G_2$ is \emph{injective (bijective)}
  if both $f_E$ and $f_V$ are injective (bijective); 
  it is an inclusion if both $f_E(e) = e$ and $f_V(v) = v$ 
  hold for all $e \in E_1$ and $v \in V_1$.
  We write $G_1 \to G_2$ or $G_2 \gets G_1$ if there is an inclusion from $G_1$ to $G_2$, 
  in which case $G_1$ is a \emph{sub-graph} of $G_2$. 
\end{definition}

To define double pushout graph transformation
and the Borrowed Context technique~\cite{EH06}, 
we will need the following constructions of hyper-graphs, 
which roughly amount to intersection and union of hyper-graphs.

\begin{definition}[Pullbacks \& pushouts of monos]
  Let $G_i = \hg[i]$ ($i \in \{0,1,2,3\}$) be hyper-graphs
  and let $G_1 \to G_3 \gets  G_2$
  be inclusions.
  The \emph{intersection of $G_1$ and $G_2$}
  is the hyper-graph $G' = (E_1 \cap E_2, V_1 \cap V_2, \ell', \cnct')$ 
  where $\ell' (e) = \ell_1(e)$ and $\cnct'(e) = \cnct_2(e)$ 
  for all $e \in E_1 \cap E_2$. 
  The \emph{pullback of $G_1 \to G_3 \gets  G_2$} 
  is the pair of inclusions 
  $G_1 \gets G'\to G_2$
  and the resulting square is a pullback square (see Figure~\ref{fig:PB-PO}).

  \begin{figure}[htb]
    \centering
    \begin{tikzpicture}
      \mnode{G_3}{0,1}
      \mnode{G_1}{0,0}
      \mnode{G_2}{1,1}
      \mnode{G'}{1,0}
      \POC{G'}{135}
      \foreach \u/\v in {G'/G_2, G'/G_1, G_1/G_3, G_2/G_3}
      {\ardraw{\u}{\v}};
    \end{tikzpicture}
    \qquad
    \begin{tikzpicture}
      \mnode{G_0}{0,1}
      \mnode{G_1}{0,0}
      \mnode{G_2}{1,1}
      \mnode{G''}{1,0}
      \POC{G''}{135}
      \foreach \u/\v in {G'/G_2, G'/G_1, G_1/G_0, G_2/G_0}
      {\ardraw{\v}{\u}};
    \end{tikzpicture}
    \caption{Pullback and pushout square}
    \label{fig:PB-PO}
  \end{figure}
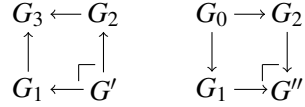

  Let  $G_1 \gets G_0 \to G_2$
  be inclusions;
  they are \emph{non-overlapping} if
  both $E_1 \cap E_2 \subseteq E_0$
  and $V_1 \cap V_2 \subseteq V_0$ hold.
  The \emph{pushout of non-overlapping inclusions $G_1 \gets G_0 \to G_2$} 
  is the pair of inclusions $G_1 \to G'' \gets G_2$ 
  where $G'' = (E_1 \cup E_2, V_1 \cup V_2, \ell'', \cnct'')$
  is the hyper-graph that satisfies
  \begin{displaymath}
    \ell''(e) =
    \begin{cases}
      \ell_1(e) & \text{ if } e \in E_1\\
      \ell_2(e) & \text{ otherwise } 
    \end{cases}
    \text{ and } 
    \cnct''(e) =
    \begin{cases}
      \cnct_1(e) & \text{ if } e \in E_1\\
      \cnct_2(e) & \text{ otherwise } 
    \end{cases}
  \end{displaymath}
  for all $e \in E_1 \cup E_2$. 
\end{definition}


Finally, we are ready to introduce graph transformation systems and 
their labelled transition semantics.

  \begin{definition}[Rules and graph transformation systems]
    A \emph{rule (scheme)} is a pair of non-overlapping inclusions of hyper-graphs
    $\rho = (L \gets I \to R)$.
    Let $A,B$ be hyper-graphs such that $A \gets L$
    and moreover\smallskip\,\,
    \begin{minipage}[c]{.8\linewidth}     
    $A \gets I \to R$ is non-overlapping.
    Now,
    \emph{$\rho$ transforms $A$ to $B$} 
    if there exists a diagram 
    as shown on the right
    such that the two squares are pushouts 
    and there is an isomorphism $\iota \colon B' \to B$.
    A \emph{graph transformation system \textsc{(gts)}}
    is pair $\S = (\Lambda, \R)$
    where $\Lambda$ is a set of labels and $\R$ is a set of rules.
    \end{minipage}%
    \begin{minipage}[c]{.2\linewidth}
        \centerline{
      \begin{tikzpicture}
        \mnode{L}{-1,0}
        \mnode{I}{0,0}
        \mnode{R}{1,0}
        \mnode{A}{-1,-1}
        \mnode{D}{0,-1}
        \mnode[B]{B'}{1,-1}
        \foreach \u/\v in {I/L,I/R,L/A,I/D,R/B,D/A,D/B}
        {
          \ardraw{\u}{\v}
        };
        \POC{B}{135}
        \POC{A}{45}
      \end{tikzpicture}
    }
    \end{minipage}%
  \end{definition}

A graph transformation rule can be understood as follows.
Whenever the left hand side $L$ is (isomorphic to) a sub-graph of some graph $A$
then this sub-graph can be ``removed'' from $A$, yielding the graph $D$. 
The vacant place in $D$ is then ``replaced'' by the right hand side $R$ of the rule.
The middleman $I$ is the memory of the connections $L$ had with the rest of the graph in order for $R$ to be attached in exactly the same place.

We now present an example that will be used throughout the paper to illustrate 
the main ideas.
\begin{example}
  The system $\S_{ex} = (\Lambda,\R)$ will be the following one in the
  sequel:
  $\Lambda = \{\alpha,\beta,\gamma, \dots \}$ such that $\ar(\alpha) =
  2$, $\ar(\beta) = 3$ and $\ar(\gamma) = 1$;
  moreover
  $\R$ is the set of rules given in Figure~\ref{fig:contrex-rules}
  where the $R_i$ represent different graphs (e.g.\ edges with labels $R_i$).
  \label{ex:runningex}
\end{example}

To keep the graphical representations clear,
all inclusions in the running example are given implicitly by the spatial arrangement of nodes and edges.


\begin{figure}[htbp]
 \centering
  \subfigure[Rule ``$\alpha/\beta$'']{
  \begin{minipage}{.4\textwidth}
    \begin{tabular}{@{}c@{}}
      \scalebox{.45}{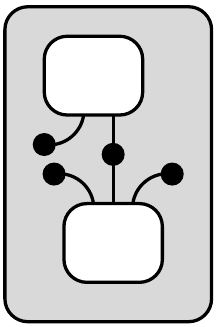}
    \end{tabular}
    $\leftarrow$
    \begin{tabular}{@{}c@{}}
      \scalebox{.45}{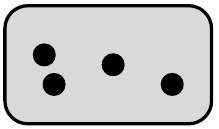} 
    \end{tabular} 
    $\rightarrow$
    \begin{tabular}{@{}c@{}}
      \scalebox{.45}{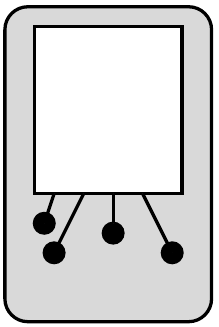}
    \end{tabular}
  \end{minipage}
  }
  \subfigure[Rule ``$\alpha/\gamma$'']{
  \begin{minipage}{.4\textwidth}
    \begin{tabular}{@{}c@{}}
      \scalebox{.45}{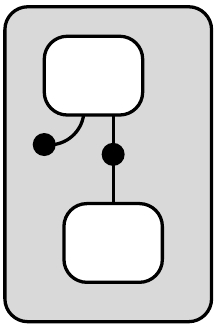}
    \end{tabular}
    $\leftarrow$
    \begin{tabular}{@{}c@{}}
      \scalebox{.45}{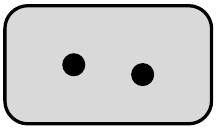} 
    \end{tabular} 
    $\rightarrow$
    \begin{tabular}{@{}c@{}}
      \scalebox{.45}{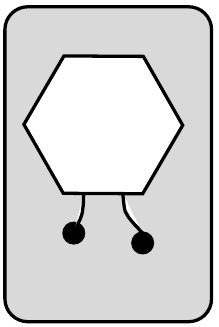}
    \end{tabular}
  \end{minipage}
  }
  \subfigure[Rule ``$\beta/\gamma$'']{
  \begin{minipage}{.4\textwidth}
    \begin{tabular}{@{}c@{}}
      \scalebox{.45}{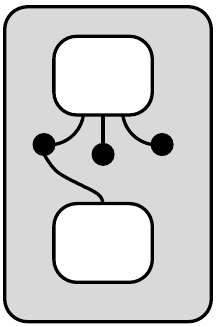}
    \end{tabular}
    $\leftarrow$
    \begin{tabular}{@{}c@{}}
      \scalebox{.45}{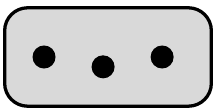} 
    \end{tabular} 
    $\rightarrow$
    \begin{tabular}{@{}c@{}}
      \scalebox{.45}{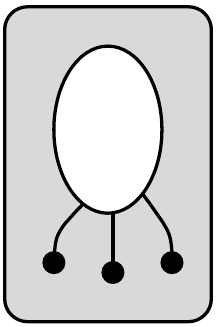}
    \end{tabular}
  \end{minipage}
  }
  \subfigure[Rule ``$\alpha/\beta/\gamma$'']{
  \begin{minipage}{.4\textwidth}
    \begin{tabular}{@{}c@{}}
      \scalebox{.45}{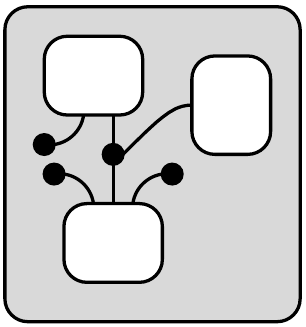}
    \end{tabular}
    $\leftarrow$
    \begin{tabular}{@{}c@{}}
      \scalebox{.45}{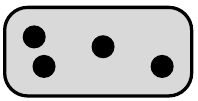} 
    \end{tabular} 
    $\rightarrow$
    \begin{tabular}{@{}c@{}}
      \scalebox{.45}{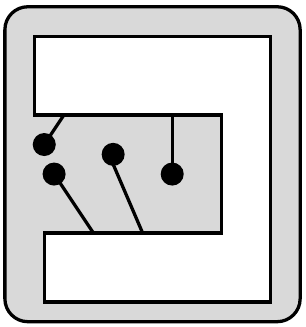}
    \end{tabular}
  \label{subfig:alfabetagamma}
  \end{minipage}
  }
 \caption{Reaction rules of $\S_{ex}$.}
 \label{fig:contrex-rules}
\end{figure}

  \begin{remark}[Rule instances]
    Given a rule $L \gets I \to R$ 
    and a graph $A$ such that $A \gets L$, 
    one can assume w.l.o.g.\ 
    that $A \gets I \to R$
    is non-overlapping.
    The reason is that in each case, 
    the rule $L \gets I \to R$
    could be replaced by an \emph{isomorphic} ``rule  instance''
    $\rho'= L' \gets I' \to R'$
    (based on the standard notion of rule isomorphism). 
  \label{rem:ruleInstance}
  \end{remark}

  \marginpar{\begin{example}[Graph Rewriting]
    [TODO, interaction nets ? ]
  \end{example}}

  In fact the result of each transformation step is unique (up to
  isomorphism). 
  This is a consequence of the following fact. \medskip

  \noindent\begin{minipage}[c]{.77\linewidth}
    \begin{fact}[Pushout complements]
      \label{fact:poc}
      Let $G_2 \gets G_1 \gets G_0$ be a pair of hyper-graph
      inclusions where $G_i = \hg[i]$ ($i \in \{0,1,2\}$) such that
      for all $v \in V_1 \setminus V_0$ there does not exist any edge
      $e \in E_2 \setminus E_0$ such that $v$ is incident to $e$.
      Then there exists a unique sub-graph $G_2 \gets D$ such that
      \eqref{eq:poc}~is a pushout square.
    \end{fact}
  \end{minipage}%
  \begin{minipage}[c]{.23\linewidth}
        \begin{equation}
        \label{eq:poc}
      \begin{tikzpicture}[baseline={(current bounding box.west)}]
        \mnode{G_1}{0,0}
        \mnode{G_0}{1,0}
        \mnode{G_2}{0,-1}
        \mnode{D}{1,-1}
        \foreach \u /\v in {G_0/G_1,G_1/G_2,D/G_2,G_0/D}
        {
          \ardraw{\u}{\v}
        };
        \POC{G_2}{45}
      \end{tikzpicture}
    \end{equation}
  \end{minipage}
  \medskip

  \begin{definition}[Pushout Complement]
    Let $G_2 \gets G_1 \gets G_0$ be a pair of hyper-graph inclusions
    that satisfy the conditions of Fact~\ref{fact:poc};
    the unique completion $G_2 \gets D \gets G_0$ in \eqref{eq:poc}
    is the \emph{pushout complement} of $G_2 \gets G_1 \gets G_0$.
  \end{definition}

\begin{definition}[Labelled transition system]
  A \emph{labelled transition system} (\textsc{lts}) is a tuple $(S, \L , R)$
  where $S$ is a set of \emph{states},
  $\L$ is a set of \emph{labels} 
  and $R \subseteq S \times \L \times S$ 
  is the \emph{transition relation}. 
  We  write 
  \begin{displaymath}
    s \xrightarrow{\alpha} s' 
  \end{displaymath}
  if $(s,\alpha,s') \in R$
  and say that $s$ can evolve to $s'$ by performing $\alpha$. 
\end{definition}
\medskip 

\noindent\begin{minipage}[c]{.7\linewidth}
\begin{definition}[DPOBC]%
  \label{def:dpobc}
  Let $\S = (\Lambda, \R)$ be a graph transformation system.
  Its \textsc{lts}
  has all inclusions of hyper-graphs $J \to G$ 
  as \emph{states} where $J$ is called the \emph{interface};
  the labels are all pairs of inclusions
  $J \to F \gets K$,
  and
  a state $J \to G$ evolves to another one $K \to H$ 
  if there is a 
  diagram as shown on the right, 
which is called a \emph{\textsc{dpobc}-diagram}  or just a \emph{\textsc{bc}-diagram}. 
In this diagram, the graph $D$ is called the \emph{partial match of $L$}.
\end{definition}
\end{minipage}%
\begin{minipage}[c]{.3\linewidth}
  \begin{center}
\begin{tikzpicture}[scale=1,baseline={(current bounding box.west)},semithick]
  \mnode[D]{D}{0,3}
  \mnode[L]{L}{1,3}
  \mnode[I]{I}{2,3}
  \mnode[R]{R}{3,3}
  \mnode[G]{G}{0,2}
  \mnode[Gc]{G_c}{1,2}
  \mnode[C]{C}{2,2}
  \mnode[H]{H}{3,2}
  \mnode[J]{J}{0,1}
  \mnode[F]{F}{1,1}
  \mnode[K]{K}{2,1}

  \foreach \u/\v in {D/L,I/L,I/R,D/G,L/Gc,I/C,R/H,G/Gc,J/G,J/F,F/Gc,K/C,R/H,C/Gc,K/F,C/H,K/H}
  {\ardrawd{\u}{\v}};
  
  \POC{Gc}{45}
  \POC{Gc}{135}
  \POC{Gc}{225}
  \POC{H}{135}
  \POC{K}{135}

\end{tikzpicture}
\end{center}
\end{minipage}\medskip

For a technical justification of this definition, see \cite{sassonesobocinski:njc}, but let us give some intuitions on what this diagram expresses.
States are inclusions, where the ``larger'' part models the whole ``internal'' state of the system while the ``smaller'' part, the interface, models the part that is directly accessible to the environment and allows for (non-trivial) interaction. As a particular simple example, one could have a Petri net where the set of places (with markings) is the complete state and some of the place are ``open'' to the environment such that interaction takes place by exchange of tokens.

The addition of agents/resources from the environment might result in
``new'' reactions, which have not been possible before. The idea of the
\textsc{lts} semantics for graph transformation is to consider (the addition
of) ``minimal'' contexts that allow for ``new'' reactions as labels. 
The minimality requirement of an addition $J \to E$ or $J \to F$ is
captured by the two leftmost squares in the BC diagram above:
the addition $J \to F$ is ``just enough'' to complete part of the left
hand side $L$ of some rule. If the reaction actually takes place,
which is captured by the other two squares in the upper row in the BC
diagram, some agents might disappear / some resources might be used
(depending on the preferred metaphor) and new ones might
appear. Finally the pullback square in the BC diagram restricts the
changes to obtain the new interface into the result state after
reaction. As different rules might result in different deletion effects
that are ``visible'' to the environment, the full label of each such
``new'' reaction is the ``trigger'' $J \to F$ together with the
``observable'' change $F \gets K$ (with state $K \to H$ after interaction).


%
%
\section{Three Layer SOS semantics}
\label{sec:SOSsemantics}
We start with a reformulation of the borrowed context technique
that breaks the ``monolithic'' \textsc{bc}-step
into axioms (that allow to derive the \emph{basic actions})
and two rules that allow to perform these basic actions within suitable contexts.
The axioms corresponds to the \textsc{ccs}-axioms
that describe that the process $\alpha.P$ can perform the action $\alpha$ 
and then behaves as~$P$, written $\alpha.P \to[\alpha] P$
where $\alpha$ ranges over the actions $a, \bar a$, and $\tau$. 
In the case of graphs, 
each rule $L \gets I \to R$
gives rise to such a set of actions. 
More precisely, 
each subgraph $D$ of $L$ can be seen as an ``action'' with 
co-action $\widehat D^L \to L$ such that $L$ is the union of $D$ and $\widehat D^L$.
For example, in the rule $\alpha/\beta$,
both edges $\alpha$ and $\beta$ yield (complementary) basic actions.

Formally, in Table~\ref{tab:transitions}, we have the family of \emph{Basic Action} axioms. 
It essentially represents all the possible uses of a transformation rule.
In an (encoding of) \textsc{ccs},
the left hand side would be a pair of unary edges $a$ and $\bar a$, 
which both disappear during reaction. 
Now, if only $a$ is present ``within'' the system, 
it needs $\bar a$ to perform a reaction; 
thus, the part $a$ of the left hand side induces the (inter-)action that consists in ``borrowing'' $\bar a$ and deleting both edges 
(and similarly for $\bar a$). 
In general, e.g.\ in the rule $\alpha/\beta/\gamma$ there might be more than
two edges that are involved in a reaction and thus we have a whole family of actions. 
More precisely, 
each portion of a left hand side induces the action that consists in borrowing the missing part to perform the reaction
(thus obtaining the coplete left hand side),
followed by applying the changes that are described by the right part of the rule.

Next, we shall give counterparts for two \textsc{ccs}-rules
that describe that an action can be performed in parallel to another process
and under a restriction. 
More precisely, 
whenever we have the transition $P \to[\alpha] P'$
and another process $Q$, 
then there is also a transition $P \parallel Q \to[\alpha] P'\parallel Q$;
similarly, 
we also have $(\nu b)P \to[\alpha] (\nu b)P'$ whenever
$\alpha\notin\{\bar b, b\}$.
More abstractly, 
actions are preserved by certain contexts. 
The notion of context in the case of graph transformation, 
which will be the counterpart of process contexts such as $P \parallel [\cdot]$ and $(\nu b) [\cdot]$,
is as follows. 
\begin{definition}[Context]
  \label{def:the-context}
  A \emph{context} is a pair of inclusions $C=\context{J}{E}{J'}$. 
  Let $J \to G$ be a state
  (such that $E \gets J \to G$ is non-overlapping);
  the \emph{combination} of $J \to G$ with the context~$C$,
  written $C[J \to G]$, 
  is the inclusion of $J'$ into the pushout 
  of $E \gets J \to G$
  as illustrated in the following display.
  \begin{displaymath}
    \text{state:}
    \begin{tikzpicture}[baseline={(0,0.5)}]
      \mnode{J}{0,0}
      \mnode{G}{0,1}
      \foreach \u/\v in {J/G}{
        \ardraw{\u}{\v}
      };
    \end{tikzpicture}
    \quad
    \text{context:}    
    \begin{tikzpicture}[baseline={(0,0.5)}]
      \mnode{J}{0,0}
      \mnode{E}{1,0}
      \mnode{J'}{2,0}
      \foreach \u/\v in {J/E,J'/E}{
        \ardraw{\u}{\v}
      };
    \end{tikzpicture}
    \quad\text{construction:}
    \begin{tikzpicture}[baseline={(0,0.5)}]
      \mnode{J}{0,0}
      \mnode{G}{0,1}
      \mnode{E}{1,0}
      \mnode{J'}{2,0}
      \mnode[PO]{\bar G}{1,1}
      \POC{PO}{-135}
      \foreach \u/\v in {J/E,J/G,E/PO,G/PO, J'/PO,J'/E}{
        \ardraw{\u}{\v}
      };
    \end{tikzpicture}
    \quad
    \text{combination:}    
    \begin{tikzpicture}[baseline={(0,0.5)}]
      \mnode[J]{J'}{0,0}
      \mnode[G]{\bar G}{0,1}
      \foreach \u/\v in {J/G}{
        \ardraw{\u}{\v}
      };
    \end{tikzpicture}
  \end{displaymath}
\end{definition}

The left inclusion of the context, i.e.\ $J \to E$, can also be seen as a state with the same interface.
The pushout then gives the result of ``gluing'' $E$ to the original $G$ at the interface $J$;
the second inclusion $J' \to E$ models a new interface,
which possibly contains part of $J$ and additional ``new'' entities in $E$.

With this general notion of context at hand, 
we shall next address the counterpart of name restriction, 
which we call \emph{interface narrowing}, 
the second rule family in Table~\ref{tab:transitions}.
In \textsc{ccs}, 
the restriction $(\nu a)$ 
preserves only those actions that do not involve $a$. 
The counterpart of the context  $(\nu a)[\cdot]$
is a context of the form  $J \to J \gets J'$. 
In certain cases, 
one can ``narrow'' a label
while ``maintaining'' the ``proper'' action
as made formal in the following definition. 

\begin{definition}[Narrowing]
  A \emph{narrowing context}
  is a context of the form $C = J \to J \gets J'$.
  Let $J \rightarrow F \leftarrow K$ be a label
  such that the pushout complement of $F \gets J \gets J'$ exists;
  then the \emph{$C$-narrowing} of the label, 
  written $C[J \rightarrow F \leftarrow K]$ 
  is the lower row in the following display
\begin{displaymath}
\begin{tikzpicture}[scale=1.1,baseline={(current bounding box.west)},semithick]
  \mnode[J']{\makebox[0pt][r]{\ensuremath{C[J \rightarrow F \leftarrow K]:= 
        {}}}J'}{0,0}
  \mnode{J}{0,1}
  \mnode{F'}{1,0}
  \mnode{F}{1,1}
  \mnode{K'}{2,0}
  \mnode{K}{2,1}
  \foreach \u/\v in {J'/F',K'/F',J/F,K/F,J'/J,F'/F,K'/K}
  {\ardrawd{\u}{\v}};
  \POC{F}{225}
  \POC{K'}{135}
  %
\end{tikzpicture}
\qquad \text{ where } C = J \to J \gets J'
\end{displaymath}
where the left square is a pushout and the right one a pullback. 
Whenever we write $C[J \rightarrow F \leftarrow K]$,
we assume that the relevant pushout complement exists. 

If we think of the interface as the set of free names of a process, 
then restricting a name means removal from the interface.
Thus,  $J'$ is the set of the remaining free names. 
If the pushout complement $F'$ exists, it represents $F$ with the restricted names erased.
Finally, since a pullback here can be seen as an intersection,
$K'$ is $K$ without the restricted names.
So we finally obtain the ``same'' label where ``irrelevant'' names are not mentioned. 
It is of course not always possible to narrow the interface.
For instance, 
one cannot restrict the names that are involved in labelled transitions of \textsc{ccs}-like process calculi. 
This impossibility is captured by the non-existence of the pushout complement.

%
\end{definition}
With the notion of narrowing, 
we can finally define the interface narrowing rule in Table~\ref{tab:transitions}.

\marginpar{Anrei convinced of this explanation.
Tobias not :) -- but nvm. } 
The final rule in Table~\ref{tab:transitions}
captures the counterpart of performing an action 
in parallel composition with another process $P$. 
In the case of graph transformation, 
this case is non-trivial
since even the pure addition of context potentially interferes
with the action of some state $J \to G$.
For example, if an interaction involves the deletion of an (isolated) node, 
the addition of an edge to this node inhibits the reaction.
However, 
for each transition there is a natural notion of non-inhibiting context; 
moreover, 
to stay close to the intuition that parallel composition 
with a process $P$ only adds new resources
and to avoid overlap with the narrowing rule,  
we restrict to monotone contexts.
\medskip

\noindent\begin{minipage}[c]{.75\linewidth}
\begin{definition}[Compatible contexts]
  \label{def:compatible-context}
  Let  $C = \context{J}{E}{\bar J}$ be a context;
  it is \emph{monotone}   if $J \to \bar J$. 
  Let $\redlabel{J}{F}{K}$
  be a label;
  now $C$ is \emph{non-inhibiting w.r.t.\ $\redlabel{J}{F}{K}$} 
  if it is possible to construct the  diagram
\eqref{eq:nonInhib}
where both squares are pushouts. 
Finally, 
a context $\context{J}{E}{\bar J}$ is \emph{compatible} with the label $\redlabel{J}{F}{K}$
if it is non-inhibiting w.r.t.\ it and monotone.
\end{definition}
\end{minipage}
\begin{minipage}[c]{.25\linewidth}
  \begin{equation}
\begin{tikzpicture}[scale=1.1,baseline={(current bounding box.west)},semithick]

  \mnode{E}{0,0}
  \mnode{J}{0,1}
  \mnode{E_1}{1,0}
  \mnode{F}{1,1}
  \mnode{E'}{2,0}
  \mnode{K}{2,1}

  \foreach \u/\v in {J/F,K/F,E/E_1,E'/E_1,J/E,F/E_1,K/E'}
  {\ardrawd{\u}{\v}};
  
  \POC{E_1}{45}
  \POC{E_1}{135}
\end{tikzpicture}
\label{eq:nonInhib}
\end{equation}
\end{minipage}
\medskip

In a label $\redlabel{J}{F}{K}$, the left inclusion represents the addition of new entities that ``trigger'' a certain reaction.
A compatible context is simply a context that is able to provide at least $F$, usually  more than $F$, 
while not attaching new edges to nodes that disappear during reaction.

The last rule in the \textsc{sosbc}-system of
Table~\ref{tab:transitions}
is the embedding of a whole transition into a monotone context. 
To define this properly, 
we introduce a partial operation for the ``combination''
of co-spans
(which happens to be a particular type of relative pushout of co-spans); 
this generalizes the narrowing construction.

\begin{definition}[Cospan combination]
Let $C= (J \to F \gets K)$ and $\bar C = (J \to E \gets \bar J)$ be
two cospans. 
They are \emph{combinable}
if there exists a diagram of the following form. 
\begin{displaymath}
  \begin{tikzpicture}[scale=1.1,baseline={(current bounding box.west)},semithick]

  \mnode{E}{0,0}
  \mnode{J}{0,1}
  \mnode{E_1}{1,0}
  \mnode{F}{1,1}
  \mnode{E'}{2,0}
  \mnode{K}{2,1}
  \mnode[Jb]{\bar J}{0,-1}
  \mnode[Fb]{\bar F}{1,-1}
  \mnode[Kb]{\bar K\makebox[0pt][l]{\ensuremath{{} =: \bar C[J \to F
        \gets K]}}}{2,-1}

  \foreach \u/\v in {J/F,K/F,E/E_1,E'/E_1,J/E,F/E_1,K/E',Jb/Fb,Kb/Fb,Jb/E,Fb/E_1,Kb/E'}
  {\ardrawd{\u}{\v}};
  
  \POC{E_1}{45}
  \POC{E_1}{135}
  \POC{E_1}{-135}
  \POC{Kb}{135}
   
\end{tikzpicture}
\end{displaymath}
The label $\bar J \to \bar F \gets \bar K$ is the \emph{combination of
  $C$ with  $\bar C$}, 
and is denoted by $\bar C[J \to F \gets K]$. 
\end{definition}

In fact, it is easy to show that compatible contexts are combinable
with their label. 
\begin{lemma}
  Given a reduction label $\redlabel{J}{F}{K}$ and a compatible context $\context{J}{E}{\bar J}$ for it, we can split the diagram $\ref{eq:nonInhib}$ in order to get

\begin{displaymath}
\begin{tikzpicture}[scale=1.1,baseline={(current bounding box.west)},semithick]

  \mnode{E}{0,0}
  \mnode{E_1}{1,0}
  \mnode{E'}{2,0}
  \mnode[barJ]{\bar J }{0,.75}
  \mnode[barF]{\bar F}{1,.75}
  \mnode[barK]{\bar K}{2,.75}
  \mnode{J}{0,1.5}
  \mnode{F}{1,1.5}
  \mnode{K}{2,1.5}

  \foreach \u/\v in {J/barJ,F/barF,K/barK,barJ/E,barF/E_1,barK/E',J/F,K/F,barJ/barF,barK/barF,E/E_1,E'/E_1}
  {\ardrawd{\u}{\v}};
  
  \POC{barF}{135}
  \POC{E_1}{135}
  \POC{K}{225}
  \POC{barK}{225}

\end{tikzpicture}
\qquad \text{ and } \qquad
\begin{tikzpicture}[scale=1.1,baseline={(current bounding box.west)},semithick]

  \mnode{E}{0,0}
  \mnode{J}{0,1}
  \mnode{E_1}{1,0}
  \mnode{F}{1,1}
  \mnode{E'}{2,0}
  \mnode{K}{2,1}
  \mnode[Jb]{\bar J}{0,-1}
  \mnode[Fb]{\bar F}{1,-1}
  \mnode[Kb]{\bar K\makebox[0pt][l]{\ensuremath{{} = \bar C[J \to F
        \gets K]}}}{2,-1}

  \foreach \u/\v in {J/F,K/F,E/E_1,E'/E_1,J/E,F/E_1,K/E',Jb/Fb,Kb/Fb,Jb/E,Fb/E_1,Kb/E'}
  {\ardrawd{\u}{\v}};
  
  \POC{E_1}{45}
  \POC{E_1}{135}
  \POC{E_1}{-135}
  \POC{Kb}{135}
   
\end{tikzpicture}. 
\end{displaymath}
\label{lem:inhibitRedLab}
\end{lemma}
With this lemma we can finally define the rule that corresponds 
to ``parallel composition'' of an action with another ``process''. 
Now the \textsc{sosbc}-system does not only give an analogy to 
the standard \textsc{sos}-semantics for \textsc{ccs}, 
we shall also see that the labels that are derived 
by the standard  \textsc{bc} technique
are exactly those labels 
that can be obtained from the basic actions 
by compatible contextualization and interface narrowing.
In technical terms, 
the \textsc{sosbc}-system of Table~\ref{tab:transitions}
is sound and complete. 

\begin{table}[!ht]
    \centering
    \hbox{\raisebox{0.4em}{\vrule depth 0pt height 0.4pt width 12cm}}
    \begin{itemize}
    \item \textbf{Basic Actions}
      \[
      \genfrac{}{}{}{0}{}{\quad \action{D}{D}{L}{I}{R} 
        \quad
      }
      \qquad \text{ where } 
        \begin{array}[c]{l}
          (L \leftarrow I \rightarrow R) \in \S\\
          \text{ and } D \to L
        \end{array}
      \]
    \item \textbf{Interface Narrowing} 
	\[\genfrac{}{}{}{0}{\theaction}{\quad 
           \action{G}{J'}{F'}{K'}{H} 
          \quad} \qquad
      \text{ where }
      \begin{array}[c]{l}
        C = J \to J \gets J' \\
        \text{ and } J' \to F' \gets K'=        C[J \to F \gets K]
      \end{array}
	\]
    \item \textbf{Compatible Contextualization} 
      \[\genfrac{}{}{}{0}{\theaction}{
        \quad 
        C[J \to G]
        \xrightarrow{C[J \to F \gets K]}
        \bar C[K \to H]
        \quad
      } 
      \qquad \text{ where } 
                \begin{array}[c]{l}
                  C = \context{J}{E}{\bar J} 
		  \text{ compatible with } \redlabel{J}{F}{K}\\
                  \text{ and }
                  \bar C = (J \to F \gets K)[C]
                \end{array}
	\]
    \end{itemize}
    \hbox{\raisebox{0.4em}{\vrule depth 0pt height 0.4pt width 12cm}}
\caption{Axioms and rules of the \textsc{sosbc}-system.}
\label{tab:transitions}
\end{table} 

\begin{theorem}[Soundness and completeness]
  Let $\S$ be a graph transformation system.
  Then there is a \textsc{bc}-transition
  \begin{displaymath}
    \theaction
  \end{displaymath}
  if and only if it is derivable in the \textsc{sosbc}-system. 
\label{th:SOSequalsBC}
\end{theorem}

The main role of this theorem is not its technical ``backbone'', 
which is similar to many other theorems on the Borrowed Context
technique. 
The main insight to be gained is
the absence of any ``real'' communication between sub-systems;
roughly, every reaction of a state can be ``localized''
and then derived from a basic action
(followed by contextualization and narrowing).
In particular, 
we do not have any counterpart to the communication-rule in
\textsc{ccs}, 
which has complementary actions $P \to[a] P'$ and $Q \to [\bar a] Q'$ 
as premises and concludes the possibility of communication of the processes
$P$ and $Q$ to perform the silent ``internal'' transition $P \parallel Q
\to[\tau] P' \parallel Q'$. 
The main goal is to provide an analysis of possible issues with
a counterpart of this rule. 


%
%
\section{The composition rule for CCS-like systems}
\label{sec:interaction-rule}

Process calculi, such as \textsc{ccs} and the $\pi$-calculus, 
have a so-called \emph{communication rule} that allows to synchronize sub-processes to
perform silent actions.
The involved process terms have complementary actions that allow to
interact 
by a ``hand-shake''.
However, it is an open question how such a communication rule can be
obtained for general graph transformations systems
via the Borrowed Context technique.
Roughly, 
the label of a transition does not contain information about which
reaction rule was used to derive it;
in fact, 
the same label might be derived using different rules. 
Intuitively, 
we do not know how to identify the  two hands that have met to shake hands.

To elaborate on this using the  metaphor of handshakes, 
assume that we have an agent that needs a hand to perform a
handshake 
or to deliver an object. 
If we observe this agent reaching out for another hand, 
we cannot conclude from it which of the two possible actions will follow. 
In general, 
even after the action is performed,
it still is not possible to know the decision of the agent --
without extra information, which might however not be observable. 
\marginpar{???! 
\\Intuitively, 
we first show that if the agent was able to shake a hand, and its
partner too, 
we could have done it, even if it is not what ``we did''.}
However, 
with suitable assumptions about the ``allowed actions'',
all necessary information might be available. 



First, we recall from~\cite{BEK06} 
that  \textsc{dpobc}-diagrams (as defined in Definition~\ref{def:dpobc}) can be composed under certain circumstances. 
\begin{fact}
  Let \[\theaction \quad\text{ and }\quad \action{G'}{J'}{F'}{K'}{H'} \] 
  be two transitions obtained from two \textsc{dpobc}-diagrams
  with the same rule $\rho = L \leftarrow I \rightarrow R$.
  Then, it is possible to build a \textsc{dpobc}-diagram with the same rule for the composition of $J \rightarrow G$ and $J' \rightarrow G'$ along some common interface $J \gets \minint{D}{L} \to J'$%
  . 
\end{fact}

Take the following example as illustration of this fact. 

\begin{example}[Composition of transitions] 
  Let $J \to G$ be a state of $\S_{ex}$ 
  that contains an edge $\alpha$ with its second connection in the interface
  as shown in Figure~\ref{subfig:st}.
  Further, let $J' \to G'$ be a state 
  that contains  an edge $\beta$ with its second connection in the interface
  as shown in Figure~\ref{subfig:nd}.
  Both graphs can trigger a reaction from rule $\alpha/\beta/\gamma$. 
  Such a composition is shown in Figure~\ref{subfig:rd}.
\end{example}

 \begin{figure}[htb]
 \centering
  \subfigure[A first transition]{
    \label{subfig:st}
      \scalebox{.45}{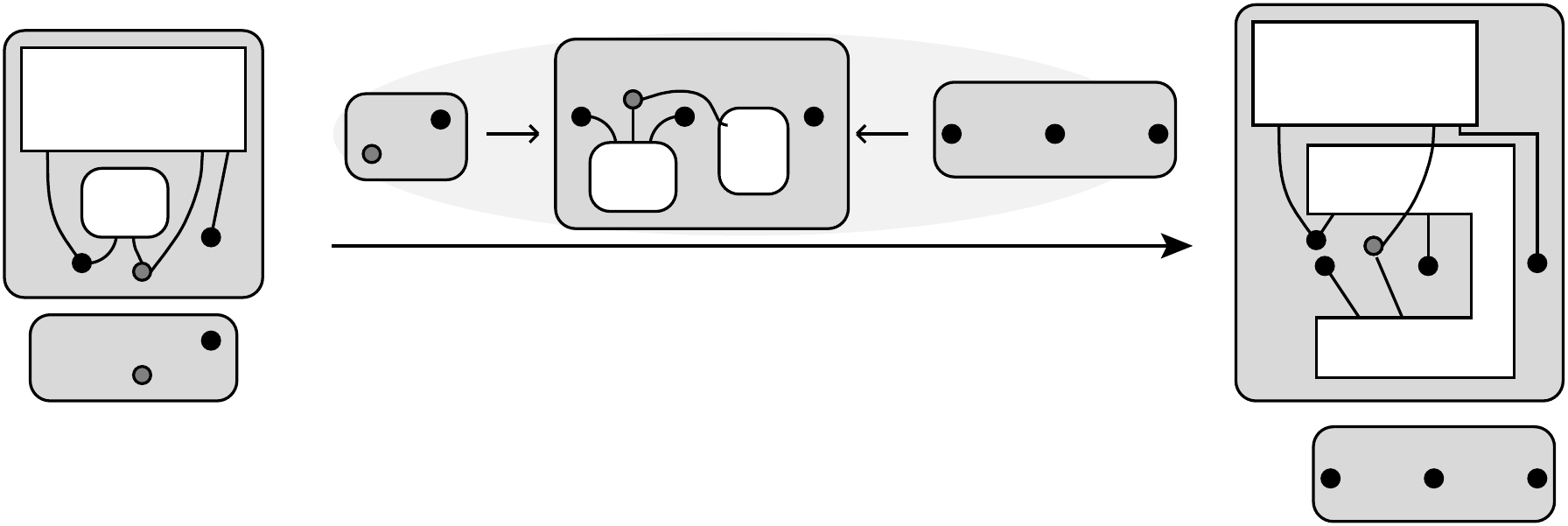}
  }
  
  \subfigure[A second transition]{
        \label{subfig:nd}
      \scalebox{.45}{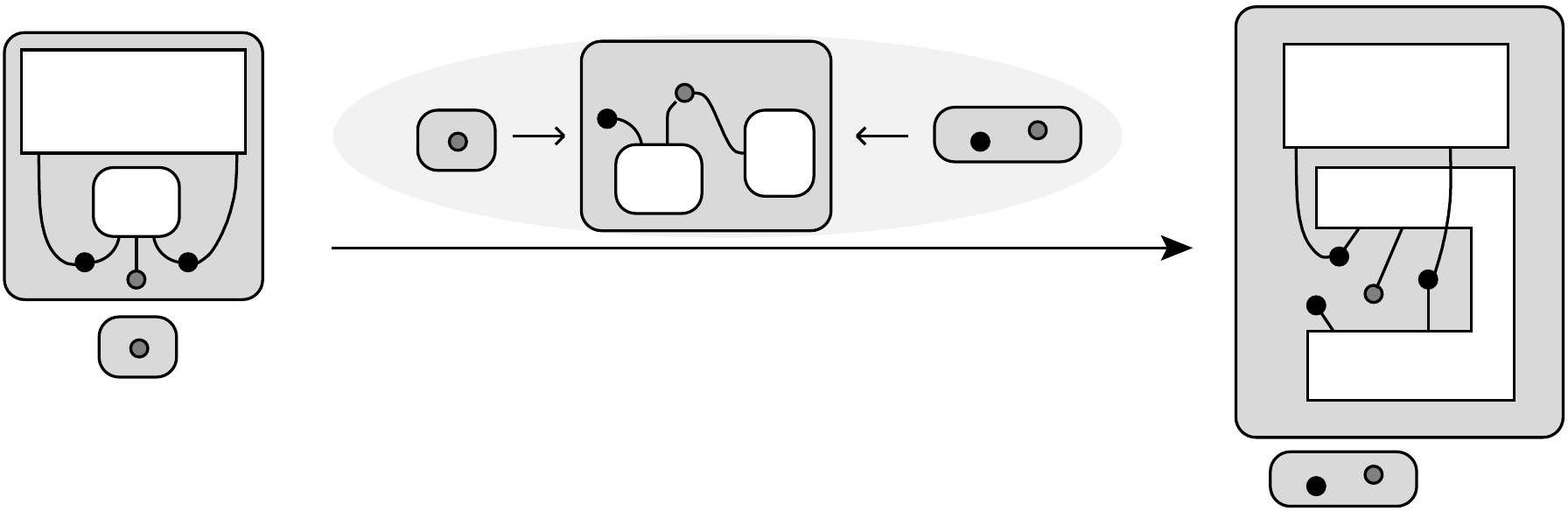}
  }
  
  \subfigure[The composition of the transitions]{
    \label{subfig:rd}
      \scalebox{.45}{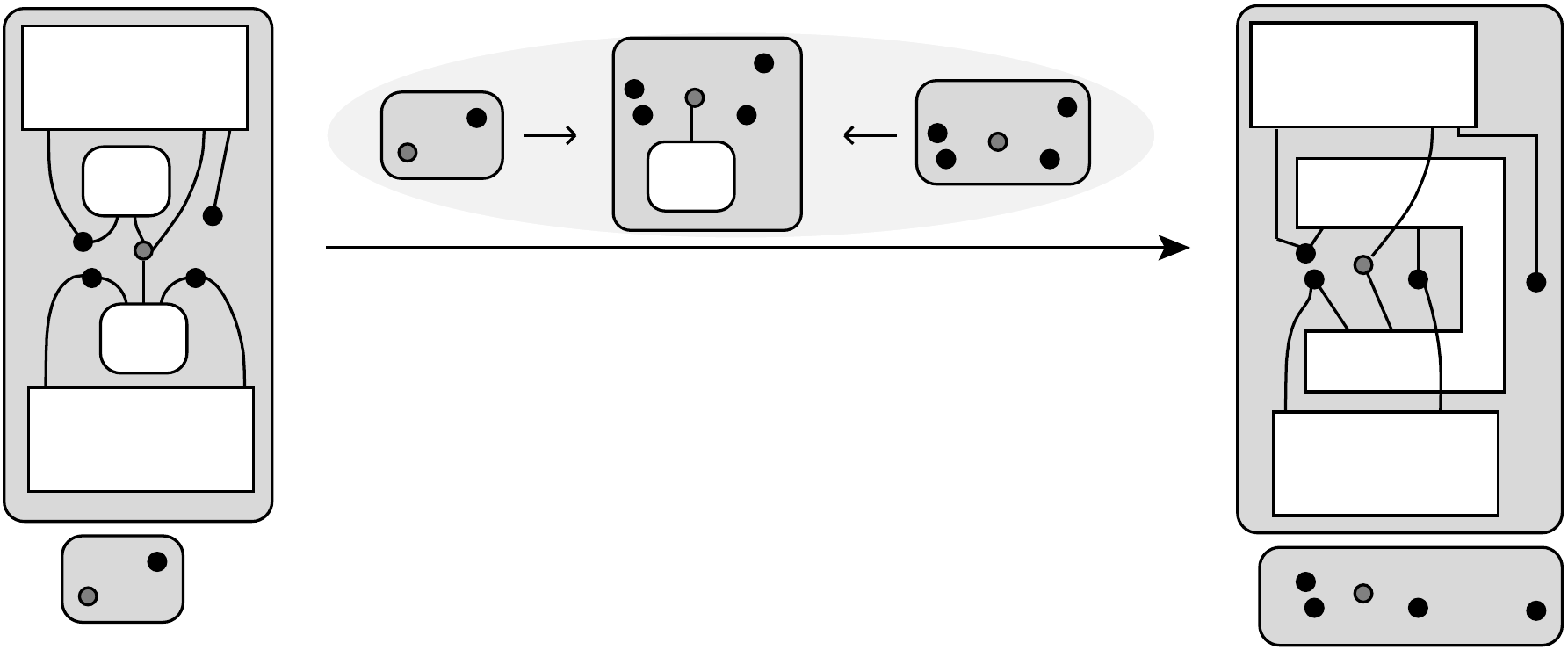}
  }
  \caption{An example of composition.}
  \label{fig:compo}
 \end{figure}
Hence, 
we see that is in general possible to combine  transitions to obtain new transitions. 
However, 
we emphasize at this point, 
that derivability of a counterpart of the communication rule of \textsc{ccs}
is not the same question as the composition 
of pairs of transitions that come equipped with \emph{complete}~\textsc{bc}-diagrams. 
To clarify the problem, 
consider the following example where we cannot  infer the used rule from the 
transition label.

\begin{example}
  Let $G$ be a graph composed of two edges $\alpha$ and $\beta$ and consider a transition label where an edge $\gamma$ is ``added''. 
  Then it is justified by both rules $\alpha/\gamma$ and $\beta/\gamma$ (see Figure \ref{fig:twoPossibilities}). 
\end{example}

 \begin{figure}[htbp]
 \centering
  \subfigure[A transition from rule $\alpha/\gamma$]{
      \scalebox{.45}{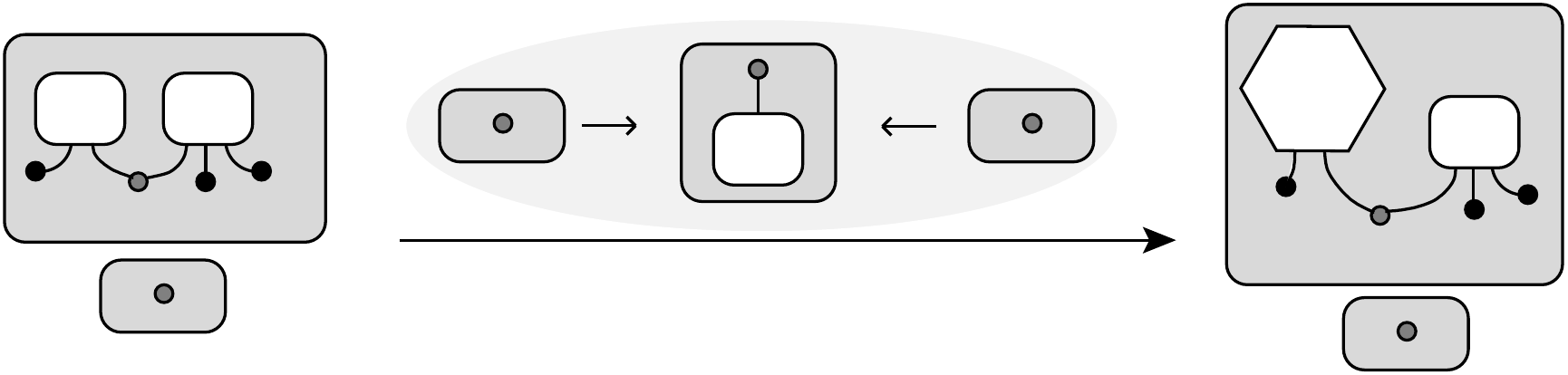}
  }
  
  \subfigure[A transition from rule $\beta/\gamma$]{
      \scalebox{.45}{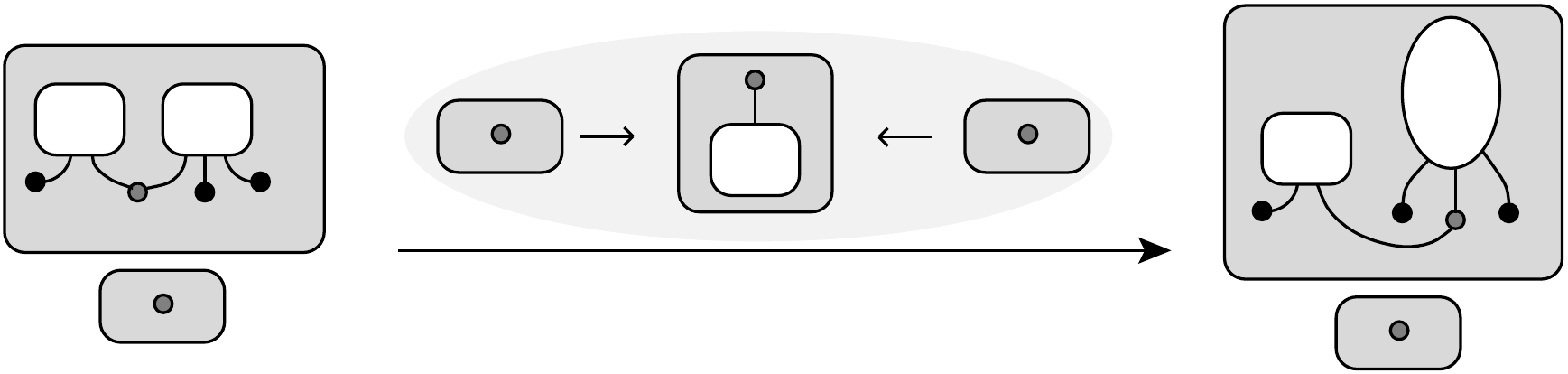}
  }
  \caption{Same transition label for different rules.}
  \label{fig:twoPossibilities}
 \end{figure}
We shall avoid this problem by restricting to suitable classes of graph transformation systems. 
Moreover, 
for simplicities sake, 
we shall focus on the derivation of ``silent'' transitions
in the spirit of the communication rule of \textsc{ccs}. 
\begin{definition}[Silent label]
  A label $J \to F \gets K$ is \emph{silent} or $\tau$
  if $J = F = K$;
  a \emph{silent transition} is a transition with a silent label. 
\end{definition}
  Intuitively, 
  a silent transition is one that does not induce any ``material''
  change that is visible to an external observer
  that only has access to  the interface of the states. 
  Hence, 
  in particular, 
  a silent transition does not involve additions of the environment
  during the transition.
  Moreover, 
  the interface remains unchanged. 
  This latter requirement does not have any counterpart 
  in process calculi, 
  as the interface is given implicitly by the set of all free names.
  (In graphical encodings of process terms~\cite{bonchi2009labelled}
  it is possible to have free names in the interface
  even though there is no corresponding input or output prefix in
  the term.)


  Now, with the focus on silent transitions,
  for a given rule
  $\redrule{L}{I}{R}$
  we can illustrate the idea of complementary actions as follows.
  If a graph $G$ contains a
  subgraph $D$ of $L$  and moreover a graph $G'$ has the complementary
  subgraph of $D$ in $L$ in it,
  then $G$ and $G'$ can be combined to obtain a
  big graph $\bar G$ -- the ``parallel composition'' of $G$ and $G'$ --
  that has the whole left hand side $L$ as a subgraph
  and thus $\bar G$ can perform the reaction.
  {A natural example for this 
    are Lafont's interaction nets
    where the left hand side consist exactly of two hyper-edges, 
    which in this case are called cells.}
  \marginpar{ so we will only consider complementary subgraphs with
    common interface a single ``wire'' (a single element of dimension
    $0$). But this can be generalized with some technicalities to any
    kind of ``minimal interface''.}
  The intuitive idea of complementary (basic) actions is 
  captured by the notion of \emph{active pairs}.

\begin{definition}[Active pairs] 
  For any inclusion $D \to L$, where $D \neq L$ and for all nodes $v$ of $D$, $\deg(v) > 0$, 
  let the following square be its initial pushout  
  \[ \begin{tikzpicture}[scale=1,baseline={(current bounding box.west)},semithick]

  \mnode[J]{\minint{D}{L}}{0,1}
  \mnode{D}{0,0}
  \mnode[D']{\compl{D}{L}}{1,1}
  \mnode{L}{1,0}

  \foreach \u/\v in {J/D,J/D',D/L,D'/L}
  {\ardrawd{\u}{\v}};

  \POC{L}{135}
\end{tikzpicture}, 
\]
i.e.\ $\compl{D}{L}$ is the smallest subgraph of $L$  that allows for completion to a pushout. 
We call $\compl{D}{L}$ the \emph{complement of $D$ in $L$} and $\minint{D}{L}$ the \emph{minimal interface of $D$ in $L$}
and we write $\{ D,D' \} \equiv L$ if $D' = \compl{D}{L}$. 
 The set of \emph{active pairs} is 
 \begin{displaymath}
   \mathbb D = \big\{\ \{D,\compl{D}{L}\}\ \mid  L \gets I \to R \in \R,D \to L,~ D \neq L,~ \forall  v \in D\ldotp \deg(v) > 0 
 \big\}.
 \end{displaymath}
  Abusing notation, 
  we also denote by $\mathbb D$ the union of $\mathbb D$.  
\label{def:activepair}
\end{definition}
It is easy to verify that the complement of $\compl{D}{L}$ in $L$ is $D$ itself and that its minimal interface is also $\minint{D}{L}$.
It is the set of ``acceptable'' partial matches in the sense that they do not yield a $\tau$-reaction on their own.
Indeed, if $D$ is equal to $L$, 
then the resulting transition of this partial match is a $\tau$-transition.
And if it is just composed of vertices, its complement is $L$
and thus not acceptable.

\begin{example}[Active pairs]
  In our running example, 
  the set $\mathbb D$ of our example is in obvious bijection to
    \[\big\{ \{\alpha,\beta\},\{\alpha,\gamma\},\{\beta,\gamma\},\{\alpha,\beta+\gamma\},\{\alpha+\beta,\gamma\},\{\alpha+\gamma,\beta\}  \big\}.\] 
  The minimal interface of any pair is a single vertex.
\end{example}
This completes the introduction of preliminary concepts to tackle the issues that 
have to be resolved to obtain ``proper'' compositionality of transitions. 

\subsection{Towards a partial solution}
Let us address the problem of identifying the rule that is ``responsible'' for a given interaction. 
We start by considering the left inclusions of labels, 
which intuitively describe possible borrowing actions from the environment.
Relative to this,  we define the \emph{admissible rules} as those rules that
can be used to let states evolve while borrowing the specified ``extra material'' from the environment.

\begin{definition}[Admissible rule]
  \label{def:admissibility}
  Let $J\to G$ be a state
  and let $J\to F$ be an inclusion (which represents a possible contribution of the context). 
  A rule $\rho$ is \emph{admissible  (for $J\to F$)}  
  if $L \not \to G$ and it is possible to find $D \in \mathbb D$ and $L$ the left-hand side of $\rho$, such that the following diagram commutes 
\begin{displaymath}
  \begin{tikzpicture}[scale=1.0,baseline={(current bounding box.west)},semithick]

  \mnode[star]{J^{^{_ L}}_{^D}}{0,0}
  \mnode{G}{0,2}
  \mnode{J}{0,1}
  \mnode{F}{1,1}
  \mnode[Gc]{G^c}{1,2}
  \mnode{D}{1,0}
  \mnode{L}{1,3}

  \foreach \u/\v in {star/J,star/D,J/G,J/F,D/F,L/Gc,G/Gc,F/Gc}
  {\ardrawd{\u}{\v}};

  \draw [->] (1.2,0) arc (-90:0:.8) -- (2,2.2) arc (0:90:.8) ;
  \draw [->] (L) -- (G) node[midway] {$\backslash$} ;
  
  \POC{Gc}{225}

\end{tikzpicture}
\end{displaymath}
where $J^{^{_ L}}_{^D} \to D$ is the minimal interface of $D$ in $L$. 
We call $D$ the \emph{rule addition}.

\end{definition}
This just means that $G$ \emph{can} evolve using the rule $\rho$ if  $D$ is added at the proper location. 

\begin{proposition}[Precompositionality]
   Let $\theactionLight$ and $\actionLight{G'}{J'}{F'}{K'}{H'}$ be two transitions such that a single rule $\rho$ is admissible  for both, and let $D$ and $D'$ be their respective rule additions.
  If $\{D,D'\} \in \mathbb D$, 
  it is possible to compose $G$ and $G'$ into a graph $\bar G$ in a way to be able to derive a
  $\tau$-transition using rule $\rho$.
\label{prop:compositionality}
\end{proposition}

\begin{proof}
  We first show that in such a case, $D' \rightarrow G$ and the pushout of $\redrule{G}{D'}{L}$ is exactly $G^c$. Similarly, $D \rightarrow G'$ and the pushout of $\redrule{G'}{D}{L}$ is exactly $G'^c$.
  Then, it is easy to see that it is possible to build the
  \textsc{dpobc}-diagram $\mathtt D_1$  using rule $\rho$  on $G$
  (respectively $G'$) yelding the transition
  $\action{G}{J}{F}{K_1}{H_1}$ for some $K_1,H_1$ (respectively the
  \textsc{dpobc}-diagram $\mathtt D_2$ yelding the transition
  $\action{G}{J}{F}{K_2}{H_2}$ for some $K_2,H_2$), and then compose
  $\mathtt D_1$ and $\mathtt D_2$.

This follows from  $\{D,D'\} \in \mathbb D$ and $\bar G \equiv \bar {G^c}$.
Indeed, $\bar E = L$ so the top left morphism of the composed \textsc{dpobc}-diagram is an isomorphism
and so are the ones under it, using basic pushout properties.
\end{proof}

This first result  motivates the following definition. 
\begin{definition}[$\tau$-compatible]
  In the situation of Proposition~\ref{prop:compositionality}, 
  we say the two transitions  are $\tau$-\emph{compatible}. 
\end{definition}

\begin{remark}
  In general, in Proposition~\ref{prop:compositionality}, 
  the result of the $\tau$-transition cannot be constructed from $H$ and $H'$;
  thus we do not yet speak of compositionality.
\end{remark}

\begin{example}
  Let $G$ be a graph composed of two edges $\alpha$ and $\gamma$ and $G'$ of two edges $\beta$ and $\gamma$  (see Figure~\ref{fig:tauComp}).
  Then the rule $\alpha/\beta$ is admissible for both transitions
  and moreover they are $\tau$-compatible. 
  The rule $\alpha/\beta$ yields the respective rule additions.
  ``Glueing'' $G$ and $G'$ by their interface results in a graph with edges $\alpha, \beta$ and two $\gamma$s;
  the latter graph  can perform a $\tau$-reaction  from rule $\alpha/\beta$,
  which however does not give the desired result
  since the target state is not the ``expected  composition'' of $H$ and $H'$. 
  In other words, 
  although we have been able to  construct a $\tau$-transition,
  it is not the composition of the original transitions.
\end{example}

 \begin{figure}[htbp]
 \centering
  \subfigure[A transition from rule $\beta/\gamma$]{
      \scalebox{.45}{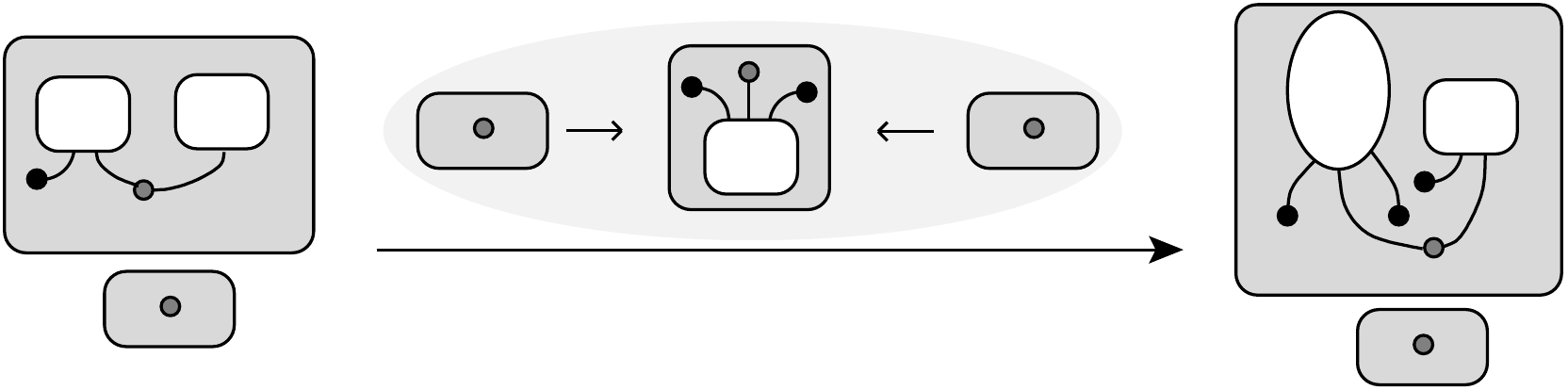}
  }
  
  \subfigure[A transition from rule $\alpha/\gamma$]{
      \scalebox{.45}{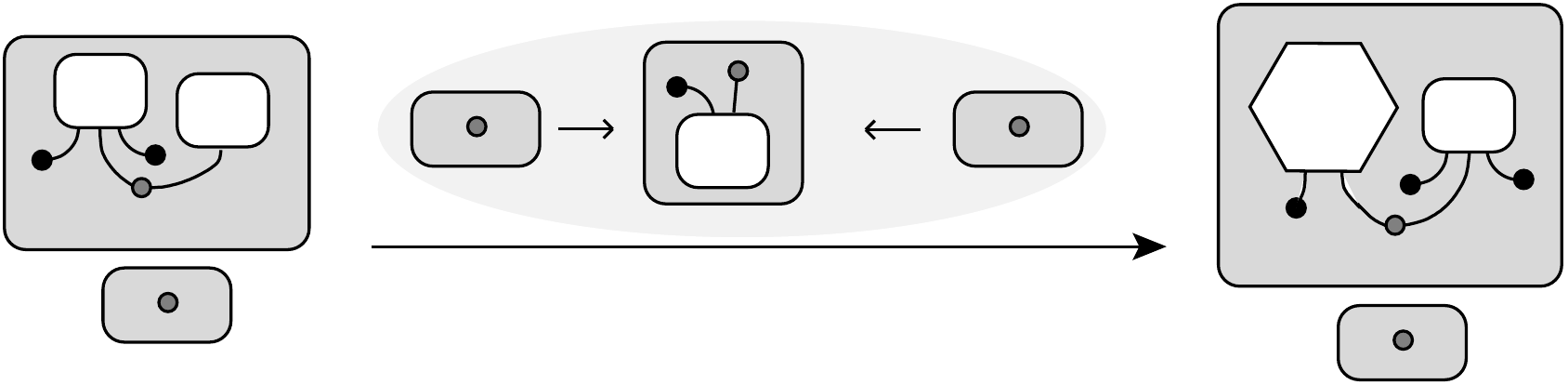}
  }
  \caption{$\tau$-compatible, but not composable: different rules.}
  \label{fig:tauComp}
 \end{figure}

We can see from the examples here that the difficulty of defining a composition of transitions comes mainly from three facts. 
The first is that a partial match can have several subgraphs triggering a reaction. 
This is delt with by the construction of the set of active pairs.
The second one is the possibility to connect multiple edges together, not knowing which one exactly is consumed in the reaction. 
Finally,  a given edge can have multiples ways of triggering a reaction.

\subsection{Sufficient conditions}
We now give two frameworks in which neither of the two last problems do occur. 
Avoiding each of them separately is enough to define compositionality properly.
Both cases  are inspired by the study of interaction net systems~\cite{Laf95,EhrReg06,MazzaPhd06},
which can be represented in the 
{}obvious{}    
manner as graph transformation systems. 
In these systems, 
the \textsc{dpobc}-diagram built from an admissible rule of a transition is necessarily the one that has to be used to derive the transition. 
In one case, 
it works for essentially the same reasons as in \textsc{ccs}: 
every active element can only interact with a unique other element, 
such as $a$ vs. $\bar a$, $b$ vs. $\bar b$.
In the other one, 
the label itself is not enough, 
but since we also know where it ``connects'' to the graph, 
it is possible to ``find'' the partner that was involved in the transition.

We introduce interaction graph systems, 
which are caracterized among other rewriting systems by the form of the left-hand sides of the reaction rules, 
composed of exactly two hyperedges connected by a single node.
We fix a labeling alphabet $\Lambda$.
\begin{definition}
  An \emph{activated pair} is a hypergraph $L$ on $\Lambda$ composed of two hyperedges $e$ and $f$ and a node  $v$ such that $v$ appears exactly once in $\cnct(e)$ and once in $\cnct(f)$.
  If $v$ is the $i$-th incident vertex of $e$ labelled $\alpha$ and the $j$-th incident vertex of $f$ labelled $\beta$, we denote the activated pair by $\pair{e_i}{f_j}$ and label it by $\pair{\alpha_i}{\beta_j}$.

  An \emph{interaction graph system} $(\Lambda,\R)$ is given by a set of reaction rules $\R$ over hypergraphs on $\Lambda$ 
  where all left-hand side of rules are activated pairs, 
  and  nodes are never deleted, 
  i.e.\ for any rule $\rho = \redrule{L}{I}{R}$, 
  \begin{itemize}
    \item $L$ is an activated pair;
    \item for any node $v$, $v \in L \Rightarrow v \in I$.
  \end{itemize}
  \label{def:interSys}
\end{definition}
Note that for any interaction graph system, 
the set $\mathbb D$ is composed of pairs $\{D,D'\}$ where each of them is composed of an edge and its connected vertices.
Also the minimal interface of any active pair $\{D,D'\}$ is a single node.
It is also the case that it is enough for interfaces to be composed of vertices only.


\begin{example}
\centerparagraph{Simply wired hypergraphs}
  Lafont interaction nets are historically the first interaction nets. 
  They appear as an abstraction of linear logic proof-nets~\cite{Laf95}.
  Originally, 
  Lafont nets have several particular features, but the one we are interested in is the condition on connectivity.
\begin{definition}
  Let $N = \hg$ be a hypergraph on $\Lambda$.

  The graph $N$ is \emph{simply wired} if $\forall v \in V$, $\deg(v) \leq 2$.
  When $\deg(v) = 1$, 
  we say that $v$ is \emph{free}.
\end{definition}

In other words, vertices are only incident to  at most two edges of a graph.
Note that in this special case 
no issues arise if we restrict to the sub-category of simply wired hypergraphs.
For this, we argue that the purpose of the interface is the possible addition of extra context; 
thus, in simply wired hypergraphs, 
it is meaningless for a vertex that is already connected to two edges to be in the interface.

\begin{definition}[Lafont interaction graph system]
  A \emph{Lafont interaction graph} is a simply connected graph such that its interface consists  of free vertices only.
  A \emph{Lafont system} $\mathbb L = (\Lambda,\R)$ is given by  reaction rules over Lafont interaction graphs;
  it is \emph{partitioned} if two left-hand sides only overlap trivially, 
  i.e.\ for two rules $\rho_j = L_j \leftarrow I_j \rightarrow R_j \in \R$ ($j=1,2$), 
  either $L_1 = L_2$ or $L_1 \cap L_2$ is the empty graph (without any nodes and any hyperedges). 
\end{definition}

\begin{lemma}
  Let $\mathbb L$ be a partitioned Lafont system, 
  let $J \rightarrow G$ be a state, let $\theaction$ be a non-$\tau$ transition.
  Then there is exactly one  admissible rule for this transition.
\label{lem:therule1}
\end{lemma}
\end{example}

\begin{example}
  \centerparagraph{Hypergraphs with unique partners}
  By generalizing Lafont interaction nets, we obtain so called \emph{multiwired} interaction nets. 
  But then we lose the unicity of the rule for a given transition label.
  It can be recovered by another condition.

  \begin{definition}[Unique partners]
    Let $\mathbb I = (\Lambda,\R)$ be an interaction graph system.
    We say it is \emph{with unique partners} if for any $\alpha \in \Lambda$ and for all $i \leq \ar(\alpha)$, 
    there exists a unique $\beta \in \Lambda$ and a unique $j \leq \ar(\beta)$ such that $\pair{\alpha_i}{\beta_j}$ is the label of a left-hand side of a rule in $\R$.
  \end{definition}

  \begin{lemma}
    Let $J \rightarrow G$ a state of\/ $\mathbb I$ and $\theaction$ a non-$\tau$ reaction label.
    Then there is exactly one  admissible rule $\rho$ for this transition.
    \label{lem:therule2}
  \end{lemma}
\end{example}


Finally, we conclude our investigation with the following  positive  result. 
\begin{theorem}[Compositionality]
  Let $(\Lambda,\R)$ be a Lafont interaction graph system, 
  or an interaction graph system with unique partners.
  Let $\mathbb D$ be its set of active pairs.
  
  Let $t_1 = \theaction$ and $t_2 = \action{G'}{J'}{F'}{K'}{H'}$ be two non-$\tau$ transitions and $D$ and $D'$  their respective rule additions.

  If $\{ D,D' \} \equiv L \in \mathbb D$,  let $\bar G$ and $\bar H$ are described by  the following  diagrams
\begin{displaymath}
\begin{tikzpicture}[scale=.8,baseline={(current bounding box.west)},semithick]
  
  \mnode[star]{\minint{D}{L}}{0,1}
  \mnode[j]{J}{1,0}
  \mnode[jj]{J'}{1,2}
  \mnode[g]{G}{3,0}
  \mnode[gg]{G'}{3,2}
  \mnode[barg]{\bar G}{4,1}
  \mnode[barj]{\bar J}{2,1}
  \POC{barg}{180}
  \POC{barj}{180}
  \foreach \u/\v in {star/j,star/jj,j/barj,jj/barj,j/g,jj/gg,star/g,star/gg,g/barg,gg/barg} 
  {\ardrawd{\u}{\v}};
  \mnode[r]{R}{6,1}
  \mnode[h]{H}{7,0}
  \mnode[hh]{H'}{7,2}
  \mnode[barh]{\bar H}{8,1}
  \POC{barh}{180}
  \foreach \u/\v in {r/h,r/hh,h/barh,hh/barh} 
  {\ardrawd{\u}{\v}};
\end{tikzpicture}
\end{displaymath}
where $\minint{D}{L} \to J$ and $\minint{D}{L} \to J'$ are the inclusions from the admissibility  of $\rho$ for states $J \to G$ and $J' \to G'$ (Definition~\ref{def:admissibility}).

  Then \[\action{\bar G}{\bar J}{\bar J}{\bar J}{\bar H}.\]
\end{theorem}
\begin{proof}[Sketch of proof]
  By Lemma~\ref{lem:therule1} or~\ref{lem:therule2}, 
  there exists exactly one rule $\rho \in \R$ with $L$ as a left-hand side that allows to derive transitions $t_1$ and $t_2$ -- 
  it is indeed the same rule for both. 
  Let $\mathtt D$ be the composition diagram of the
  \textsc{dpobc}-diagrams justifying the transitions.

  It is first shown that $\bar G \equiv \bar G_c$. 
  Since the upper and lower left squares of $\mathtt D$ are pushouts
  we can infer that  $\bar D \equiv L$ and $\bar J \equiv \bar F$.
  Finally, since no vertex is deleted (see Definition~\ref{def:interSys}),
  we have  $\bar J \to \bar C$ and thus $\bar K \equiv \bar J$.

  So $\mathtt D$ is a \textsc{bc}-diagram of a $\tau$-reaction from $\bar J \to \bar G$ to $\bar J \to \bar H$.

\end{proof}


In fact, the main property that we have used is the following. 

\noindent\begin{minipage}[c]{.67\linewidth}\begin{definition}[Complementarity of Actions]
  \label{def:complementarityOfActions}
  A graph transformation systems satisfies 
  \emph{Complementarity of Actions}
  if for each transition $\theaction$
  there is a unique rule $L\gets I \to R$
  such that there exists a \textsc{dpobc}-diagram
as shown to the right.
\end{definition}
\end{minipage}
\begin{minipage}[c]{.33\linewidth}
    \begin{displaymath}
    \begin{tikzpicture}[scale=1,baseline={(current bounding box.west)},semithick]
  \mnode[D]{D}{0,3}
  \mnode[L]{L}{1,3}
  \mnode[I]{I}{2,3}
  \mnode[R]{R}{3,3}
  \mnode[G]{G}{0,2}
  \mnode[Gc]{G_c}{1,2}
  \mnode[C]{C}{2,2}
  \mnode[H]{H}{3,2}
  \mnode[J]{J}{0,1}
  \mnode[F]{F}{1,1}
  \mnode[K]{K}{2,1}

  \foreach \u/\v in {D/L,I/L,I/R,D/G,L/Gc,I/C,R/H,G/Gc,J/G,J/F,F/Gc,K/C,R/H,C/Gc,K/F,C/H,K/H}
  {\ardrawd{\u}{\v}};
  
  \POC{Gc}{45}
  \POC{Gc}{135}
  \POC{Gc}{225}
  \POC{H}{135}
  \POC{K}{135}
\end{tikzpicture}
  \end{displaymath}
\end{minipage}
\medskip

In this situation, 
we can effectively determine if two transitions are $\tau$-compatible. 
Thus we can derive a counterpart of the communication rule of \textsc{ccs}. 
Hence, 
if a graph transformation systems satisfies 
{Complementarity of Actions}
then a rule of the  following form is derivable in \textsc{sosbc}.
\begin{displaymath}
  \frac
  {\textstyle t=\fullaction{G}{\bar J}{\bar F}{\bar K}{H} \qquad t'= \fullaction{G'}{\bar J}{\bar F'}{\bar K'}{H'}}
  {\textstyle \fullaction{\bar G}{\bar J}{\bar J}{\bar J}{\bar H}}
  \quad
  \text{ $t$ and $t'$ $\tau$-compatible }
\end{displaymath}
In other words,
in a graph transformation system with Complementarity of Actions
we can apply the results of~\cite{BEK06} 
to obtain a counterpart to the communication rule. 

\section{Related and Future work}
On a very general level, 
the present work is 
meant to strengthen the conceptual similarity of graph transformation systems
and process calculi;
thus it is part of a high-level research program 
that has been the theme of a Dagstuhl Seminar in
2005~\cite{knig_et_al:DSP:2005:27}.
In this wide field,
structural operational semantics is
occasionally considered as an instance of 
the tile model (see~\cite{DBLP:conf/birthday/GadducciM00} for an
overview). 
With this interpretation, \textsc{sos} has served as motivation for work on operational semantics 
of graph transformation systems
(e.g.~\cite{DBLP:conf/icalp/CorradiniHM00}).

A new perspective on operational semantics, 
namely the ``automatic'' generation of labeled transition semantics
from reaction rules, 
has been provided by the seminal work 
of Leifer and Milner~\cite{DBLP:conf/concur/LeiferM00}
and its successors~\cite{sassonesobocinski:njc,EH06};
as an example application, we want to mention the ``canonical''  operational semantics for the
ambient calculus~\cite{Rathke2010}. 
The main point of the latter work is 
the focus on the  ``properly''
\emph{inductive} definition of structural operational semantics.
To the best of our knowledge, 
there is no recent work on the operational semantics of graph
transformation systems
that provides a general method for the inductive definition of \emph{operational
semantics}.  
This is not to be confused with the inductive definition of graphical
encodings of process calculi on (global) states. 

With this narrower perspective on techniques for the ``automatic''
generation of \textsc{lts}s, 
we want to mention that some ideas of our three layer semantics in
Section~\ref{sec:SOSsemantics} can already be found
in~\cite{bonchi2009labelled}, 
where all rules of the definition of the labelled transition semantics
have at most one premise.
This is in contrast to the work of~\cite{Rathke2010}
where the labelled transition semantics is derived from two smaller subsystems:
the process view and the context view;
the subsystems are combined to obtain the operational semantics. 
The latter work is term based 
and it manipulates complete subterms of processes 
using the lambda calculus in the meta-language. 
We conjecture that the use of this abstraction  mechanism 
is due to the term structure of processes.

Concerning future work, 
the first  extension of the theory concerns more general (hierarchical) graph-like structures 
as captured by adhesive categories~\cite{adhesivejournal}
and their generalizations (e.g.~\cite{braatz2010finitary}). 
Moreover, 
as an orthogonal development, 
we plan to consider the case 
of more general rules that are allowed to have an  arbitrary (graph) morphism on the right hand 
side;
moreover,  also states are arbitrary morphisms. 
The general rule format is important to model substitution in name passing calculi
while arbitrary graph morphisms as states yield more natural
representations of (multi-wire) interaction nets.
The main challenge is the quest for more general sufficient conditions
that allow for non-trivial compositions of labelled transitions,
which can be seen as a general counterpart of the \textsc{ccs} communication rule.



\section{Conclusion}
\label{sec:conclusion}
We have reformulated the \textsc{bc} technique 
as the \textsc{sosbc}-system in Table~\ref{tab:transitions}
to make a general analogy to
the \textsc{sos}-rules  for \textsc{ccs}. 
There is no need for a counterpart of the communication rule. 
We conjecture that this is due to the ``flat'' structure of graphs
as opposed to the tree structure of~\textsc{ccs}-terms. 

The main contribution concerns questions about the derivability 
of a counterpart of the communication rule. 
First, we give an example, which illustrates that 
the derivability of such a rule is non-trivial;
however, 
it is derivable 
if the relevant graph transformation system
satisfies \emph{Complementarity of Actions}. 
We have given two classes of examples that satisfy this requirement, 
namely hyper-graphs with unique partners and simply wired hyper-graphs. 
This is a first step towards a ``properly'' inductive definition of structural operational 
semantics for graph transformation systems. 


\paragraph{Acknowledgements}
{\small We would like to thank Barbara K{\"o}nig, Filippo Bonchi 
and Paolo Baldan for providing us with drafts 
and ideas about a more general research program on 
compositionality in graph transformation. 
We are also grateful for the constructive criticism 
and the helpful comments of the anonymous referees. }

 \bibliographystyle{eptcs}
 \bibliography{SOSbib}
\end{document}